\documentclass[conference,10pt]{IEEEtran}

\IEEEoverridecommandlockouts

\usepackage[caption=false,font=normalsize,labelfont=sf,textfont=sf]{subfig}
\usepackage{textcomp}
\usepackage{stfloats}
\usepackage{verbatim}
\usepackage{cite}
\usepackage{amsmath,amssymb,amsfonts,amsthm}
\usepackage{nicefrac}
\usepackage{graphicx}
\usepackage{float}
\usepackage{textcomp}
\usepackage{xcolor}
\usepackage{bbm}
\usepackage{array}
\usepackage{rotating}
\usepackage{gensymb}
\usepackage{comment}
\usepackage[acronym]{glossaries}
\usepackage{multirow}
\usepackage{lscape}
\usepackage{makecell}
\usepackage[inline]{enumitem}
\usepackage[noend]{algpseudocode}
\usepackage{algorithm}
\usepackage{siunitx}

\usepackage{tikz}
\usepackage{tikz-3dplot}
\usetikzlibrary{math}
\usetikzlibrary{backgrounds} 
\usetikzlibrary{positioning, shapes.geometric}
\usetikzlibrary{arrows.meta}
\usetikzlibrary{patterns}
\usetikzlibrary{decorations.pathmorphing}
\usepackage{pgfplots}
\pgfplotsset{compat=1.16}
\usepgfplotslibrary{patchplots}
\usepgfplotslibrary{fillbetween}
\usepgfplotslibrary{groupplots}
\definecolor{myblue}{rgb}{0.00000,0.44700,0.74100}%
\definecolor{mygreen}{rgb}{0.46600,0.67400,0.18800}%
\definecolor{amber}{rgb}{1.0, 0.75, 0.0}
\definecolor{lightgray}{rgb}{0.83, 0.83, 0.83}
\definecolor{darkgray}{rgb}{0.4, 0.4, 0.4}
\definecolor{myred}{rgb}{0.77, 0.12, 0.23}

\usepackage{hyperref}
\usepackage{cleveref}

\usepackage{adjustbox}
\usepackage{tikz}
\usetikzlibrary{shapes,arrows}
\usetikzlibrary{positioning}
\usetikzlibrary{decorations.text}
\usetikzlibrary{external}
\newtheorem{assumption}{Assumption}
\crefname{assumption}{Assumption}{Assumptions}

\newtheorem{proposition}{Proposition}
\crefname{proposition}{Proposition}{Propositions}

\newtheorem{corollary}{Corollary}
\crefname{corollary}{Corollary}{Corollaries}

\renewcommand{\t}{\text{t}}
\renewcommand{\r}{\text{r}}

\newcommand{\pa}[1]{\left(#1\right)}
\newcommand{\abs}[1]{\left|#1\right|}

\newcommand{\acc}[1]{\left\{#1\right\}}

\newcommand{\argmax}{\text{argmax}}
\newif 
\iffullproof
\fullprooftrue
\def\BibTeX{{\rm B\kern-.05em{\sc i\kern-.025em b}\kern-.08em
    T\kern-.1667em\lower.7ex\hbox{E}\kern-.125emX}}

\begin{document}

\title{Performance Bounds for Near-Field Multi-Antenna Range Estimation of Extended Targets }

\author{\IEEEauthorblockN{Guillaume Thiran, François De Saint Moulin, Claude Oestges, Luc Vandendorpe\thanks{Guillaume Thiran is a Research Fellow of the Fonds de la Recherche Scientifique - FNRS.}\thanks{The authors thank Laurence Defraigne for fruitful discussions on \Cref{sec:amb_funct}.}}\\
\IEEEauthorblockA{ICTEAM, UCLouvain - Louvain-la-Neuve, Belgium\\}
\{guillaume.thiran,francois.desaintmoulin,claude.oestges,luc.vandendorpe\}@uclouvain.be}

\maketitle

\thispagestyle{empty}%to add vs delete page numbers: plain vs empty
\pagestyle{empty}

%\pagenumbering{gobble}

%%%%%%%%%%%%%%%%%%%%%%

\begin{abstract}
    In this paper, performance bounds for the multi-antenna near-field range estimation of extended targets are provided. First, analytic expressions of the ambiguity functions are obtained, emphasising the cooperation between the waveform delay and the near-field  phase shift information. The impact of estimating the range of an extended target with a point target model is analysed, showing that a model mismatch leads to severe performance degradation in the near-field region. Secondly, Cramér-Rao bounds are derived. Expressions emphasising the impact of various parameters are obtained, these parameters including the carrier frequency, the central frequency of the waveform, and its root-mean-square bandwidth. The near-field range information is shown to depend on the root-mean-square value of the propagation delay derivatives, this value scaling with the fourth power of the ratio between the antenna array dimension and the target range.
\end{abstract}

\begin{IEEEkeywords}
Near-field, range detection, extended targets, ambiguity function, Cramér-Rao bound
\end{IEEEkeywords}

%%%%%%%%%%%%%%%%%%%%%%

\section{Introduction}
\label{sec:introduction}

Wireless communication models usually consider Far-Field (FF) propagation, which amounts to approximating the spherical wavefront by a plane one at large distances. However, emerging wireless networks target higher carrier frequencies, lower propagation distances, and larger antenna arrays. This implies that the wavefront curvature must be accounted for with Near-Field (NF) propagation models. This NF effect can then be leveraged, providing significant gains in positioning applications.  

The NF effect can be exploited when an active point source is located in the NF region of an antenna array. In that case, the array can perform source positioning based on the received signals \cite{begriche2014exact}. As no control is exerted on the signal, this latter is either modelled as a parametric or stochastic narrowband signal \cite{el2010conditional}, with performance depending on the quantities measured by the receiver \cite{chen2023cramer}. The NF effect can also be leveraged in radar applications, in which signals are reflected by passive targets. Performance bounds are derived in \cite{hua2023near} for 3D localisation. OFDM signals are exploited in \cite{sakhnini2022near}, the localisation performance being demonstrated experimentally and evaluated with the Cramér-Rao Bound (CRB). Continuous waveforms are studied in \cite{wang2024cramer}, in which the associated CRB is computed, and the parameter impact emphasised. 

The above works consider dimensionless Point Targets (PTs), which re-emit the impinging signals isotropically. Yet, when large targets are close to the arrays, such a model might be too simplistic, as noted in \cite{sakhnini2022near}. Such large targets are named Extended Targets (ETs), and call for specific signal processing techniques. In \cite{moulin2024near}, we have designed an electromagnetic-based propagation model and the associated range estimator. A qualitative performance evaluation has been carried out numerically, calling for the associated theoretical bounds.

Therefore, this paper provides analytic performance bounds for the range estimation of ETs. On the one hand, closed-form ambiguity function expressions are provided, highlighting the cooperation between the NF phase shifts and the classical waveform delay information. Moreover, the impact of a model mismatch between the signal reflected by an ET, and the range detection through a PT model is analysed. It is shown that such a mismatch leads to a severe performance degradation. On the other hand, the CRB on the range estimation is derived. The obtained expression reveals the impact of the parameters, namely the carrier frequency, and the central frequency and Root-Mean-Square (RMS) bandwidth of the waveform. The NF effect is shown to depend on the RMS value of the propagation distance derivatives, scaling with the fourth power of the ratio between the array dimension and the target range. Analysing the CRB, an expression of the range at which the NF effect becomes insignificant with respect to the waveform information is obtained. 

Below, the system model is described in \Cref{sec:system_model}, while the ambiguity function expressions are provided in \Cref{sec:amb_funct}, and the CRBs in \Cref{sec:CRB}.

%\section{System Model Old}
%\label{sec:system_model_old}
%\input{system_model_old}

\section{System Model}
\label{sec:system_model}
%General extended target model

%r_lpl in extended or point target model

%MIMO and SIMO cases
\begin{figure}
    \centering
    \includegraphics{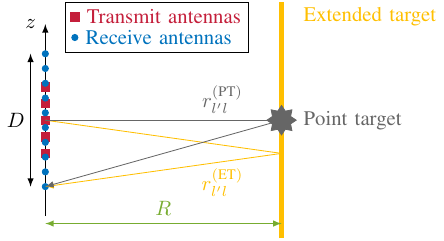}
    \vspace{-0.3cm}
    \caption{System model}
    \vspace{-0.3cm}
    \label{fig:sys_model}    
\end{figure}
The scenario of \Cref{fig:sys_model} is considered, in which a large metallic reflector is located in front of two linear antenna arrays of respectively $N_{\t}$ and $N_{\r}$ antennas, whose maximum dimension is denoted $D$. The reflector is parallel to the arrays and located at a distance $R$, named the target range. The $z$-axis position of the $l$th transmit and receive antennas is respectively denoted $ z_{l}$ and $y_{l}$. In order to estimate the range $R$, each transmit antenna emits the same waveform $s\pa{t}$ on orthogonal time resources, in a Time Division Multiple Access (TDMA) manner. TDMA indeed leads to an easier implementation than systems in which orthogonality is achieved in other domains, such as frequencies or codes \cite{xu2021transmit}. The baseband waveform is characterised by an autocorrelation
function $C\pa{\tau}\triangleq \int s\pa{t}s^*\pa{t-\tau}\mathrm{d}t$, its energy being denoted $\mathcal{E}_c \triangleq C\pa{0}$. Its frequency content ranges from $-B/2$ to $B/2$ with $B$ the bandwidth, and its central frequency and RMS bandwidth are denoted and defined as 
\begin{align}
    f_M &\triangleq \frac{1}{\mathcal{E}_c} \int f \abs{S\pa{f}}^2 \mathrm{d}f,
    \label{eq:fM}\\
    B_{\text{RMS}} &\triangleq  \sqrt{\frac{1}{\mathcal{E}_c} \int \pa{f-f_M}^2 \abs{S\pa{f}}^2 \mathrm{d}f},
    \label{eq:BRMS}
\end{align}
with $S\pa{f}$ the Fourier transform of $s(t)$. At the $l'$th receive antenna, $N_{\t}$ noisy signals are received, the $l$th one reading as 
\begin{align}
    &u_{l'l}\pa{t} = \xi \mu_{l'l;R}\pa{t} + w_{l'l}\pa{t},    \label{eq:received_signal}\\
    &\text{with} \quad \mu_{l'l;R}\pa{t} \triangleq e^{-jkr_{l'l}} s\pa{t- \frac{r_{l'l}}{c}}.
    \label{eq:model}
\end{align}
In the above, $\xi$ is an unknown complex coefficient, $k= 2\pi/\lambda$ with the wavelength $\lambda = c/f_c$ and the carrier frequency $f_c$, $c$ is the speed of light, $r_{l'l}$ is the propagation distance between the $l$th transmit and $l'$th receive antennas, and $w_{l'l}\pa{t}$ is an additive white Gaussian noise with power spectral density $\gamma_n$. In the signal model $\mu_{l'l;R}\pa{t}$, the waveform delay provides the classical range information while the range-dependent phase shifts are specific to the NF propagation. The path-loss is considered independent of the antenna locations and thus included in $\xi$. This is valid for ranges above the so-called \textit{uniform power distance}, which is within a constant factor of the array dimension $D$ \cite{lu2021communicating}. Regarding the propagation distances $r_{l'l}$, they depend on the target model, as represented in \Cref{fig:sys_model}. For a PT located at $z=0$, this distance reads as 
\begin{align}
    r_{l'l}^{\text{PT}}\triangleq \sqrt{R^2+z_l^2}+\sqrt{R^2+y_{l'}^2}.
    \label{eq:dist_PT}
\end{align}
Considering a large metallic ET, the whole target reflects the impinging wave. Yet, we have shown in \cite{moulin2024near} that for a parallel planar plate\footnote{Other PT locations or ET geometries would lead to different distance expressions. In the following, the focus is made on the two target models of \Cref{fig:sys_model}, and the acronyms PT and ET will refer to these specific instances of PTs and ETs.}, the contributions from each part of the plate combine as if the wave were solely reflected by the specular point, this latter being different for each pair of antennas. In that case \cite{moulin2024near}
\begin{align}
    r_{l'l}^{\text{ET}}\triangleq \sqrt{4R^2+\pa{z_l-y_{l'}}^2}.
    \label{eq:dist_ET}
\end{align} 
The Maximum Likelihood (ML) range estimator associated with \eqref{eq:received_signal} reads as 
\begin{align}
    \hat{R},\hat{\xi} = \argmax_{\rho,\Tilde{\xi}} \: T_{U}\pa{\acc{u_{l'l}\pa{t}}_{l,l',t}\middle|\rho,\Tilde{\xi}},
    \label{eq:range_estimation}
\end{align}
with $T_{U}\pa{\acc{u_{l'l}\pa{t}}_{l,l',t}\middle|\rho,\xi}$ the likelihood function of the received signals $u_{l'l}\pa{t}$ for all antenna pairs $\pa{l,l'}$ and all times $t$, for $R=\rho$ and $\xi = \Tilde{\xi}$. Eliminating $\xi$, the range estimator is given by $\hat{R} = \argmax_{\rho} \:\Lambda\pa{\rho}$, with \cite{moulin2024near}
\begin{align}
    \Lambda\pa{\rho} \triangleq \frac{\abs{\sum_{l'=1}^{N_{\r}}\sum_{l=1}^{N_{\t}}\int_t u_{l'l}\pa{t} \mu_{l'l;\rho}^{\star}\pa{t}\mathrm{d}t  }}{\sqrt{\sum_{l'=1}^{N_{\r}}\sum_{l=1}^{N_{\t}}\int_t\abs{\mu_{l'l;\rho}\pa{t}}^2\mathrm{d}t}}. 
    \label{eq:ML_estimator}
\end{align}
To assess the performance of this estimator, the below section studies the associated ambiguity function, and the CRB is developed in \Cref{sec:CRB}. 

\section{Ambiguity Function}
\label{sec:amb_funct}

In order to obtain the normalised ambiguity functions, \eqref{eq:ML_estimator} is evaluated for noise-free signals $u_{l'l}\pa{t} = \xi \mu_{l'l;R}\pa{t}$, and then normalised to a unit maximum. The case of targets very close to the antenna arrays being difficult to handle, two assumptions are introduced, with the Rayleigh distance $R_D\triangleq 2D^2/\lambda$. 
\begin{assumption}
    The ranges satisfy $R\geq 1.2 D$ {\normalfont \cite{bjornson2021primer}}.
    \label{ass:RgeqD}
\end{assumption}
\begin{assumption}
    \label{ass:B_fc}
    The ranges satisfy $R\geq R_D \frac{B}{f_c}$. 
    \label{ass:small_Bfc_ratio}
  \end{assumption}
Under \Cref{ass:RgeqD}, the distance $r_{l'l}$ is well approximated by its second-order Taylor series, denoted $\Tilde{r}_{l'l}$ and reading as
\begin{align}
    \Tilde{r}_{l'l}^{\text{PT}}= 2R + \frac{z_l^2+y_{l'}^2}{2R},
    &\quad 
     \Tilde{r}_{l'l}^{\text{ET}}= 2R + \frac{\pa{z_l-y_{l'}}^2}{4R},
    \label{eq:Taylor_dist}    
\end{align}
these approximations being named the Fresnel approximations.  The following proposition gives an expression of the NF range ambiguity function as a function of the classical waveform normalised ambiguity function  $\chi_C(\rho-R)\triangleq |C(2 (\rho-R)/c)/\mathcal{E}_c|$, and of the phase shift ambiguity function defined as
\begin{align} 
    \chi_{\Delta \phi}\pa{\rho,R} \triangleq  \abs{\frac{1}{N_{\r} N_{\t}}\sum_{l'=1}^{N_{\r}}\sum_{l=1}^{N_{\t}} e^{jk\pa{\tilde{\rho}_{l'l}-\tilde{r}_{l'l}}} }.
    \label{eq:NF_amb}
\end{align}  
The notations $\rho_{l'l}$ and $\Tilde{\rho}_{l'l}$ respectively refer to the propagation distances \eqref{eq:dist_PT} and \eqref{eq:dist_ET}, and approximate distances \eqref{eq:Taylor_dist}, for a range equal to $\rho$ instead of $R$.
\begin{proposition}
    The normalised ambiguity function of \eqref{eq:ML_estimator} is
    \begin{align}
        \normalfont
        \chi\pa{\rho,R} = \abs{\frac{1}{N_{\r} N_{\t}}\sum_{l'=1}^{N_{\r}}\sum_{l=1}^{N_{\t}} e^{jk\pa{\rho_{l'l}-r_{l'l}}} \frac{C\pa{\frac{\rho_{l'l}-r_{l'l}}{c}}}{\mathcal{E}_c}}.
        \label{eq:ambiguity_fct}
    \end{align}
    Under \Cref{ass:RgeqD,ass:small_Bfc_ratio}, \eqref{eq:ambiguity_fct} is approximated as
    \begin{align}
        \normalfont \chi\pa{\rho,R} \approx \chi_C\pa{\rho-R} \chi_{\Delta \phi}\pa{\rho,R}.
        \label{eq:amb_product}
    \end{align}
\end{proposition}
\begin{proof}
    By developing \eqref{eq:ML_estimator} for noise-free signals, one obtains $\Lambda\pa{\rho}= \abs{\xi} \sqrt{\mathcal{E}_c N_{\r} N_{\t}} \chi\pa{\rho,R}$ with $\chi\pa{\rho,R}\leq \chi\pa{R,R} = 1$. This thus proves \eqref{eq:ambiguity_fct}. Under \Cref{ass:small_Bfc_ratio}, $C(r_{l'l}/c)\approx C(2R/c)$, since the difference $(r_{l'l}-2R)/c$ remains lower than a fraction of $1/B$. This enables to have the waveform terms independent of $l$ and $l'$,  leading to the desired result.
\end{proof}
% In order to obtain the ambiguity functions, \eqref{eq:ML_estimator} is evaluated for noise free signals $u_{l'l}\pa{t} = \xi \mu_{l'l;R}\pa{t}$. Letting $\rho_{l'l}$ be the propagation distance associated with the antennas $l$ and $l'$, and the range $\rho$, $\Lambda\pa{\rho}$ is developed as
% \begin{align}
%     \Lambda\pa{\rho}%&=\frac{\abs{\xi}^2\abs{\sum_{l'=1}^{N_{\r}}\sum_{l=1}^{N_{\t}}\int_t \mu_{l'l;R}\pa{t} \mu_{l'l;\rho}^{\star}\pa{t}\mathrm{d}t}^2}{\sum_{l'=1}^{N_{\r}}\sum_{l=1}^{N_{\t}}\int_t \abs{s\pa{t- \frac{\rho_{l'l}}{c}}}^2 \mathrm{d}t}, \nonumber\\
%     %& \frac{\abs{\xi}^2\abs{\sum_{l'=1}^{N_{\r}}\sum_{l=1}^{N_{\t}} e^{-jk\pa{r_{l'l}-\rho_{l'l}}} \int_t s\pa{t- \frac{r_{l'l}}{c}} s^*\pa{t- \frac{\rho_{l'l}}{c}}\mathrm{d}t}^2}{N_{\r} \sum_{l=1}^{N_{\t}}C}, \\
%     %& \frac{\abs{\xi}^2\abs{\sum_{l'=1}^{N_{\r}}\sum_{l=1}^{N_{\t}} e^{jk\pa{\rho_{l'l}-r_{l'l}}} C\pa{\frac{\rho_{l'l}-r_{l'l}}{c}}}^2}{N_{\r} N_{\t} C},\\
%     &= \abs{\xi}^2 C N_{\r} N_{\t} \chi\pa{\rho,R}, \label{eq:ML_fct}
% \end{align}
% with the range ambiguity function
% \begin{align}
%     \chi\pa{\rho,R} \triangleq \abs{\frac{1}{N_{\r} N_{\t}}\sum_{l'=1}^{N_{\r}}\sum_{l=1}^{N_{\t}} e^{jk\pa{\rho_{l'l}-r_{l'l}}} \frac{C\pa{\frac{\rho_{l'l}-r_{l'l}}{c}}}{C}}^2.
% \end{align}
% It can be noted that this latter has a unit maximum at $\rho = R$. 

%In \eqref{eq:ML_fct}, the array gain $N_{\r} N_{\t}$ can be observed, while $\abs{\xi}^2C$ is the received power, including the path loss. 
The product of \eqref{eq:amb_product} implies that the global ambiguity function has a main lobe narrower than the one of $\chi_C$ and $\chi_{\Delta \phi}$ alone. It thus appears that the waveform and the phase information act collaboratively to provide the range estimation.   %Determining for which parameters the NF ambiguity function becomes dominant with respect to the classical waveform information will be the subject of \Cref{sec:discussion}. 

\subsection{Point and extended target models}
While \eqref{eq:ambiguity_fct} and \eqref{eq:NF_amb} can be evaluated for different antenna configurations, the impact of the parameters remains elusive. Therefore, the NF ambiguity function is analysed in this section for two antenna configurations. The first configuration is with a single transmit antenna located at $z=0$ and a receive array of dimension $D$, centred at $z=0$. In the second case, the transmit and receive arrays have the same dimension $D$, are centred at $z=0$ and have possibly different numbers of antennas. The transmit antennas emits in a MIMO-TDMA manner. These two cases will be referred to with the SIMO and MIMO acronyms in the following. In these configurations, \Cref{cor:amb_fct} shows that the NF effect is captured by the parameter $\beta$ defined as
\begin{align}
 \beta \triangleq \sqrt{R_{D}\abs{\frac{1}{\rho}-\frac{1}{R}}} = \sqrt{\frac{R_{D}}{R}}\sqrt{\frac{\abs{R-\rho}}{\rho}},
 \label{eq:beta}
\end{align}
this being reminiscent of the performance of beamfocusing in NF communications \cite{cui2022channel}. It also shows that under \Cref{ass:N_ant}, analytic functions can be obtained in terms of the Fresnel integral $F\pa{x} \triangleq \int_{0}^x e^{j\frac{\pi}{2}t^2}\mathrm{d}t$ and cardinal sine  $\normalfont \text{sinc}\pa{x} \triangleq \sin(\pi x)/(\pi x)$.
\begin{assumption}
    \label{ass:N_ant}
    The array size and the numbers of antennas are such that the summations on the antennas of the arrays can be replaced by integrals {\normalfont \cite{cui2022channel}}.
\end{assumption}

\begin{corollary}
    \label{cor:amb_fct}
    The phase shift ambiguity functions \eqref{eq:NF_amb} read as 
    \begin{align*}
        \normalfont \chi_{\Delta \phi}^{\text{PT,SIMO}}\pa{\rho,R} &= \normalfont \chi_{\Delta \phi,\r}\pa{\beta},\\
        \normalfont \chi_{\Delta \phi}^{\text{PT,MIMO}}\pa{\rho,R} & =\normalfont \chi_{\Delta \phi,\r}\pa{\beta}\chi_{\Delta \phi,\t}\pa{\beta},\\
        \normalfont\chi_{\Delta \phi}^{\text{ET,SIMO}}\pa{\rho,R} &\normalfont=  \chi_{\Delta \phi,\r}\pa{\frac{\beta}{\sqrt{2}}},\\
        \normalfont\chi_{\Delta \phi}^{\text{ET,MIMO}}\pa{\rho,R} &\normalfont =  \chi_{\Delta \phi,\r\t}\pa{\frac{\beta}{\sqrt{2}}},
    \end{align*}
    with 
    \begin{align*}
        \begin{array}{rl}
            \normalfont \chi_{\Delta \phi,\t}\pa{\beta} &\normalfont \triangleq   \abs{\frac{1}{N_{\t}} \sum_{l=1}^{N_{\t}} e^{j\frac{\pi}{2} \pa{\frac{\beta z_l}{D}}^2 }}\approx \abs{\frac{2}{\beta}F\pa{\frac{\beta}{2}}},\\
            \normalfont\chi_{\Delta \phi,\r}\pa{\beta} &\normalfont\triangleq   \abs{\frac{1}{N_{\r}} \sum_{l'=1}^{N_{\r}} e^{j\frac{\pi}{2} \pa{\frac{\beta y_{l'}}{D}}^2 }}\approx \abs{\frac{2}{\beta}F\pa{\frac{\beta}{2}}},\\
            \normalfont\chi_{\Delta \phi,\r\t}\pa{\beta} &\triangleq \normalfont \abs{\frac{1}{N_{\t} N_{\r}} \sum_{l=1}^{N_{\t}}\sum_{l'=1}^{N_{\r}} e^{j\frac{\pi}{2} \pa{\beta \frac{z_l-y_{l'}}{D}}^2 }},\\
            &\normalfont \approx \abs{2\frac{F\pa{\beta}}{\beta}  - e^{j\pi\frac{\beta^2}{4}} \text{sinc}\pa{\frac{\beta^2}{4}}},
        \end{array}
    \end{align*}
    the approximations holding under \Cref{ass:N_ant}.
\end{corollary}
\begin{proof}
    This follows from the NF ambiguity function definition \eqref{eq:NF_amb}, and the fact that under \Cref{ass:N_ant}, the summations are approximated by integrals.
\end{proof}
\begin{figure}[h]
    \centering
    \includegraphics[width=\columnwidth]{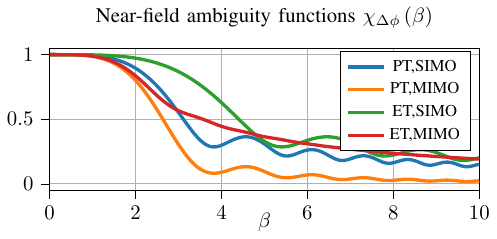}
    \vspace{-0.5cm}
    \caption{NF ambiguity function of \Cref{cor:amb_fct} as a function of $\beta$.  }
    \label{fig:NF_amb}
\end{figure}
The NF ambiguity functions, as a function of $\beta$, are depicted in \Cref{fig:NF_amb}. They are decreasing with $\beta$, possibly with some oscillations. From \eqref{eq:beta}, it appears that the range resolution increases with $R_D/R$, i.e. the estimation performance improves at small ranges, for which the NF effect is significant.

% Furthermore, the above can be approximated as 
% \begin{align}
%     \chi_{\Delta \phi,\r\t}\pa{\beta} &\approx  \abs{\frac{1}{D^2} \int_{\frac{-D}{2}}^{\frac{D}{2}}\int_{\frac{-D}{2}}^{\frac{D}{2}} e^{j\frac{\pi}{2} \pa{\beta \frac{z-y}{D}}^2 }\mathrm{d}y\mathrm{d}z}^2, \nonumber \\
%     %&=  \abs{\frac{1}{\beta^2} \int_{\frac{-\beta}{2}}^{\frac{\beta}{2}}\int_{\frac{-\beta}{2}}^{\frac{\beta}{2}} e^{j\frac{\pi}{2} \pa{t-s}^2 }\mathrm{d}s\mathrm{d}t}^2,\\
%     &=  \abs{2\frac{F\pa{\beta}}{\beta}  - e^{j\pi\frac{\beta^2}{4}} \text{sinc}\pa{\frac{\beta^2}{4}}}^2,
%     \label{eq:amb_rt_ET}
% \end{align}
% with $\text{sinc}\pa{x} \triangleq \sin(\pi x)/(\pi x)$.

\subsection{Impact of a model mismatch}
This section investigates the effect of using the PT model when the target is a large plane ET, focusing on the SIMO case. This amounts to analysing \eqref{eq:NF_amb} with $\Tilde{\rho}^{\text{PT}}_{l'l}$ and $\Tilde{r}^{\text{ET}}_{l'l}$, as given by \eqref{eq:Taylor_dist}. \Cref{cor:mismatch} provides the associated ambiguity function, its proof being similar to the one of \Cref{cor:amb_fct}.

\begin{corollary}
    \label{cor:mismatch} When the range of an ET is estimated with a PT model, the ambiguity function, in the SIMO case, is given by 
    \begin{align*}
        \normalfont\chi^{\text{Mis,SIMO}}\pa{\rho,R} &\normalfont\approx \chi_C\pa{\rho-R} \chi_{\Delta \phi,\r}\pa{\frac{\gamma}{\sqrt{2}}},
    \end{align*}
with $\gamma \triangleq \sqrt{R_{D}\abs{\frac{2}{\rho}-\frac{1}{R}}}$.
\end{corollary}
This corollary shows that in the cases of a model mismatch, the ambiguity function has a similar form as those of \Cref{cor:amb_fct}, yet with the parameter $\beta$ replaced by $\gamma$. The NF ambiguity function has then its maximum at $\rho = 2R$ instead of $R$, and thus it provides erroneous information. This is explained by the presence of $\xi$, which erases the absolute phase information. Thus, only phase differences between antennas can be measured. Using \eqref{eq:Taylor_dist}, the phase difference between the signal received at the antenna $l'$ and $l''$ is given by $\Delta_{\phi}^{\text{PT}} = k(y_{l'}^2-y_{l''}^2) /2R$ and $\Delta_{\phi}^{\text{ET}} = k(y_{l'}^2-y_{l''}^2)/4R$ for a PT and ET, $z_l$ being $0$ in the SIMO case. For the same phase difference, the estimated range will thus be affected by a factor $2$ if the erroneous model is considered. Contrarily, the waveform delay of each antenna pair can be directly measured. For a given delay, the range estimation will mainly be obtained from the $2R$ terms in \eqref{eq:Taylor_dist} as it is much larger than the term in $\mathcal{O}(1/R)$. The $2R$ terms being identical for both target models, the waveform ambiguity function is not impacted by the model mismatch. This shows that the target model may have an enormous impact on the range estimation in the NF region. However, in the FF region, $\gamma \to 0$, and thus the NF ambiguity function becomes almost constant, the target model having hence no impact on the estimation performance. 

%This can be understood by looking at \Cref{fig:mismatch}, showing the phase difference between the signal received at the center antenna, and the one located at $y_{l'} = D/2$ for the two target models. As it can be observed, and as it can be checked analytically with \eqref{eq:Taylor_dist}, for the same measured phase difference, the PT and ET model predict a range that is different by a factor of $2$.
% \begin{figure}[h]
%     \centering
%     \resizebox{\columnwidth}{!}{\input{Images/mismatch.tikz}}
%     \caption{Phase difference between the signal received at the center antenna, and the one located at $y_{l'} = D/2$, in the SIMO case. }
%     \label{fig:mismatch}
% \end{figure}

\section{Cramér-Rao bounds}
\label{sec:CRB}
In this section, the CRBs are developed in \Cref{prop:CRB}, without relying on \Cref{ass:RgeqD,ass:small_Bfc_ratio}, this being in contrast with the ambiguity function results. 

\begin{proposition}
    \label{prop:CRB}
    The CRB on the range estimation is given by
    \begin{align}
        \normalfont \text{CRB}_{R}& \normalfont= \frac{\pa{N_{\r} N_{\t}  \text{SNR} }^{-1}}{\frac{32\pi^2}{c^2} \pa{\pa{\eta - \beta^2}\pa{f_c+f_M}^2 +\eta B_{\text{RMS}}^2 }},
        \label{eq:CRB}
    \end{align}
    with $\normalfont \text{SNR} \triangleq \frac{\abs{\xi}^2 \mathcal{E}_c}{\gamma_n}$, the $\eta$ and $\beta$ factors defined as 
    \begin{align}
        \eta &\normalfont \triangleq \frac{1}{ N_{\r} N_{\t}}\sum_{l'=1}^{N_{\r}}\sum_{l=1}^{N_{\t}} \pa{\frac{1}{2}\frac{\partial r_{l'l}}{\partial R}}^2,\\
        \beta &\normalfont  \triangleq \frac{1}{ N_{\r} N_{\t}}\sum_{l'=1}^{N_{\r}}\sum_{l=1}^{N_{\t}} \pa{\frac{1}{2}\frac{\partial r_{l'l}}{\partial R}},
    \end{align}
    and with $\normalfont f_M$ and $\normalfont B_{\text{RMS}}$ defined in  \eqref{eq:fM} and \eqref{eq:BRMS}.
\end{proposition}
\begin{proof}
\iffullproof
Associated with the estimation problem \eqref{eq:range_estimation}, the Fisher information matrix is given by 
\begin{align*}
    I = -\mathbb{E}\{ \begin{bmatrix} \frac{\partial^2 \log T_U}{\partial \xi_{\mathcal{R}}^2} & \frac{\partial^2\log T_U}{\partial \xi_{\mathcal{R}}\partial \xi_{\mathcal{I}}} & \frac{\partial^2\log T_U}{\partial \xi_{\mathcal{R}}\partial R}\\
        \frac{\partial^2 \log T_U}{\partial \xi_{\mathcal{I}}\partial \xi_{\mathcal{R}}} & \frac{\partial^2\log T_U}{\partial \xi_{\mathcal{I}}^2} & \frac{\partial^2 \log T_U}{\partial \xi_{\mathcal{I}}\partial R} \\
        \frac{\partial^2 \log T_U}{\partial R\partial \xi_{\mathcal{R}}} & \frac{\partial^2\log T_U}{\partial R\partial \xi_{\mathcal{I}}} & \frac{\partial^2 \log T_U}{\partial R^2}
    \end{bmatrix}\},
    %\label{eq:CRB}
\end{align*}
with $\xi_{\mathcal{R}}$ and $\xi_{\mathcal{I}}$ the real and imaginary parts of $\xi$, and $\mathbb{E}\{\cdot\}$ the expectation operator. The log-likelihood function is developed as 
\begin{align*}
    \log T_U &= K-\frac{1}{\gamma_n }\sum_{l'=1}^{N_{\r}}\sum_{l=1}^{N_{\t}}\int_t \abs{u_{l'l}\pa{t}-\xi \mu_{l'l;R}\pa{t}}^2\mathrm{d}t,\\
    %&= \Tilde{K} -\frac{\abs{\xi}^2}{\gamma_n }\sum_{l'=1}^{N_{\r}}\sum_{l=1}^{N_{\t}}\int_t \abs{\mu_{l'l;R}\pa{t}}^2\mathrm{d}t,\nonumber\\
    %&\quad + \frac{2}{\gamma_n }\mathcal{R}\acc{\xi^{\star} \sum_{l'=1}^{N_{\r}}\sum_{l=1}^{N_{\t}}\int_t u_{l'l}\pa{t} \mu_{l'l;R}^{\star}\pa{t}\mathrm{d}t},\\
    %&= \Tilde{K} -\frac{\abs{\xi}^2 N_{\r} N_{\t} C}{\gamma_n } + \frac{2}{\gamma_n }\mathcal{R}\acc{\xi^{\star} \sum_{l'=1}^{N_{\r}}\sum_{l=1}^{N_{\t}}\int_t u_{l'l}\pa{t} \mu_{l'l;R}^{\star}\pa{t}\mathrm{d}t},\nonumber\\
    &= \Tilde{K} -\frac{\xi_{\mathcal{R}}^2 N_{\r} N_{\t} \mathcal{E}_c}{\gamma_n }-\frac{\xi_{\mathcal{I}}^2 N_{\r} N_{\t} \mathcal{E}_c}{\gamma_n }\nonumber\\
    &  \quad+ \frac{2 \xi_{\mathcal{R}}}{\gamma_n  }\mathcal{R}\acc{\sum_{l'=1}^{N_{\r}}\sum_{l=1}^{N_{\t}}\int_t u_{l'l}\pa{t} \mu_{l'l;R}^{\star}\pa{t}\mathrm{d}t}\nonumber\\
    &  \quad+ \frac{2 \xi_{\mathcal{I}}}{\gamma_n  }\mathcal{I}\acc{\sum_{l'=1}^{N_{\r}}\sum_{l=1}^{N_{\t}}\int_t u_{l'l}\pa{t} \mu_{l'l;R}^{\star}\pa{t}\mathrm{d}t}.
\end{align*}
This leads to the following elements of the Fisher information matrix:
\begin{align*}
    I =\frac{8N_{\r} N_{\t} \mathcal{E}_c}{\gamma_n }  \begin{bmatrix} \frac{1}{4} & 0 & -\frac{\mathcal{R}\acc{\xi N}}{2}\\
        0 & \frac{1}{4} &  -\frac{\mathcal{I}\acc{\xi N}}{2} \\
        -\frac{\mathcal{R}\acc{\xi N}}{2} &  -\frac{\mathcal{I}\acc{\xi N}}{2} & -\abs{\xi}^2\mathcal{R}\acc{M}
    \end{bmatrix},
\end{align*}
with $N$ and $M$ defined as
%\begin{align*}
%    I_{11} &=  I_{22} = \frac{2 N_{\r} N_{\t} C}{\gamma_n },\\
%     I_{12} &= I_{21} = 0,\\
%     I_{13} &= I_{31} =  \frac{-2}{\gamma_n  }\mathcal{R}\acc{\xi N},\\
%     I_{23} &= I_{32} =  \frac{-2}{\gamma_n  }\mathcal{I}\acc{\xi N},\\
%    I_{33} &= \frac{-2\abs{\xi}^2 }{\gamma_n  }\mathcal{R}\acc{M}.
% \end{align*}
\begin{align*}
    N &= \frac{1}{2 N_{\r} N_{\t} \mathcal{E}_c}\sum_{l'=1}^{N_{\r}}\sum_{l=1}^{N_{\t}}\int_t \mu_{l'l;R}\pa{t} {\frac{\partial \mu^{\star}_{l'l;R}}{\partial R}}\pa{t}\mathrm{d}t,\\
   M &\triangleq \frac{1}{4 N_{\r} N_{\t} \mathcal{E}_c} \sum_{l'=1}^{N_{\r}}\sum_{l=1}^{N_{\t}}\int_t \mu_{l'l;R}\pa{t}{\frac{\partial^2\mu^{\star}_{l'l;R}}{\partial R^2}}\pa{t}\mathrm{d}t.
\end{align*}
Partitionning the matrix and using the block-matrix inversion lemma \cite[Fact 2.17.3]{bernstein2009matrix}, the CRB is obtained as
\begin{align*}
    \text{CRB}_{R}
    & = \frac{1}{\frac{8\abs{\xi}^2 \mathcal{E}_c N_{\r} N_{\t} }{\gamma_n} \pa{-\mathcal{R}\acc{M}-\abs{N}^2}}.
\end{align*}
It remains to develop the expressions of $M$ and $N$ for the model \eqref{eq:model}. To that aim, the following notations are introduced:
\begin{align*}
    X &\triangleq \int_t s\pa{t} \frac{\partial s^{\star}\pa{t}}{\partial t}\mathrm{d}t,\\
    B &\triangleq \int_t \abs{\frac{\partial s\pa{t}}{\partial t}}^2 \mathrm{d}t = -\int_t s\pa{t} \frac{\partial s^{\star}\pa{t}}{\partial t^2}\mathrm{d}t,
\end{align*}
the last equality being obtained by integration by part. From the same argument, it can be observed $X$ is purely imaginary: $X=j\mathcal{I}\acc{X}$, while $B$ is real by definition. Moreover, from Parseval theorem, 
\begin{align}
    \frac{jX}{\mathcal{E}_c} &= 2\pi f_M, \quad 
   \frac{B}{\mathcal{E}_c}-\frac{|X|^2}{\mathcal{E}_c^2} = 4\pi^2 B^2_{\text{RMS}}.
\end{align}
 The following is furthermore introduced:
\begin{align*}
    \gamma &\triangleq \frac{1}{2 N_{\r} N_{\t}}\sum_{l'=1}^{N_{\r}}\sum_{l=1}^{N_{\t}} \frac{\partial^2 r_{l'l}}{\partial R^2}.
\end{align*}
Equipped with these notations, the model derivatives are obtained as
\begin{align*}
    &\frac{\partial \mu_{l'l;R}\pa{t}}{\partial R} \\ &= -e^{-jkr_{l'l}}\frac{\partial r_{l'l}}{\partial R}\pa{jk s\pa{t-\frac{r_{l'l}}{c}}+\frac{1}{c}\frac{s\pa{t-\frac{r_{l'l}}{c}}}{\partial t}},\\
    &\frac{\partial^2 \mu_{l'l;R}\pa{t}}{\partial R^2} \\ &= e^{-jkr_{l'l}}\left[\pa{-k^2\pa{\frac{\partial r_{l'l}}{\partial R}}^2-jk \frac{\partial^2 r_{l'l}}{\partial R^2} }s\pa{t-\frac{r_{l'l}}{c}}  \right.\\
    &+ \pa{\frac{2jk}{c}\pa{\frac{\partial r_{l'l}}{\partial R}}^2 - \frac{1}{c} \frac{\partial^2 r_{l'l}}{\partial R^2} }\frac{\partial s\pa{t-\frac{r_{l'l}}{c}}}{\partial t} \\
    &\left. + \frac{1}{c^2}\pa{\frac{\partial r_{l'l}}{\partial R}}^2 \frac{\partial^2 s\pa{t-\frac{r_{l'l}}{c}}}{\partial t^2} \right].
\end{align*}
This leads to
\begin{align*}
    &\int_t \mu_{l'l;R}\pa{t}\frac{\partial^2\mu^{\star}_{l'l;R}}{\partial R^2}\pa{t}\mathrm{d}t \nonumber \\
    &= \pa{-k^2\pa{\frac{\partial r_{l'l}}{\partial R}}^2+jk \frac{\partial^2 r_{l'l}}{\partial R^2} }\mathcal{E}_c\\
    &+ \pa{\frac{-2jk}{c}\pa{\frac{\partial r_{l'l}}{\partial R}}^2 - \frac{1}{c} \frac{\partial^2 r_{l'l}}{\partial R^2} }X- \frac{1}{c^2}\pa{\frac{\partial r_{l'l}}{\partial R}}^2 B.\nonumber
\end{align*}
Plugging the above into $M$,
\begin{align*}
    &\mathcal{R}\acc{M}\\
    &= \mathcal{R}\acc{-k^2\eta+2jk \gamma + \pa{\frac{2k}{c}\eta - \frac{2j}{c} \gamma }\frac{\mathcal{I}\acc{X}}{\mathcal{E}_c} - \frac{\eta}{c^2} \frac{B}{\mathcal{E}_c}},\\
    & = -k^2\eta + \frac{2k}{c}\eta\frac{\mathcal{I}\acc{X}}{\mathcal{E}_c} - \frac{1}{c^2}\eta \frac{B}{\mathcal{E}_c}. 
\end{align*}
For the second term, it comes that
\begin{align*}
    &\int_t \mu_{l'l;R}\pa{t}\frac{\partial\mu^{\star}_{l'l;R}}{\partial R}\pa{t}\mathrm{d}t = j\pa{k \mathcal{E}_c-\frac{1}{c}\mathcal{I}\acc{X}}\frac{\partial r_{l'l}}{\partial R},
\end{align*}
leading to
\begin{align*}
    &N= j\beta\pa{k -\frac{1}{c}\frac{\mathcal{I}\acc{X}}{\mathcal{E}_c}}.
\end{align*}
Therefore, 
\begin{align*}
    &-\mathcal{R}\acc{M}-\abs{N}^2\nonumber\\
&= \pa{\eta - \beta^2}\pa{k-\frac{1}{c}\frac{\mathcal{I}\acc{X}}{\mathcal{E}_c}}^2 +\frac{\eta}{c^2}\pa{\frac{B}{\mathcal{E}_c}-\frac{\abs{X}^2}{\mathcal{E}_c^2}},\\
&= \pa{\eta - \beta^2}\frac{4\pi^2}{c^2}\pa{f_c+f_M}^2 +\eta\frac{4\pi^2}{c^2}B_{\text{RMS}}^2,
\end{align*}
concluding the proof.
\else
 This result is obtained by computing the $3\times3$ Fisher Information Matrix (FIM) for $R$, and the real and imaginary parts of $\xi$, the log-likelihood of the estimation problem \eqref{eq:range_estimation} reading as 
\begin{align*}
    \log T_U &= K-\frac{1}{\gamma_n }\sum_{l'=1}^{N_{\r}}\sum_{l=1}^{N_{\t}}\int \abs{u_{l'l}\pa{t}-\xi \mu_{l'l;R}\pa{t}}^2\mathrm{d}t,
\end{align*}
with $K$ a constant. Partitioning the FIM and using the block-matrix inversion lemma, the CRB is obtained as
\begin{align*}
    \text{CRB}_{R}
    & = \frac{\gamma_n}{8\abs{\xi}^2 \mathcal{E}_c N_{\r} N_{\t}  \pa{-\mathcal{R}\acc{M}-\abs{N}^2}},
\end{align*}
with $\mathcal{R}\acc{\cdot}$ the real part operator, and $N$ and $M$ defined as
\begin{align*}
    N &\triangleq \frac{1}{2 N_{\r} N_{\t} \mathcal{E}_c}\sum_{l'=1}^{N_{\r}}\sum_{l=1}^{N_{\t}}\int \mu_{l'l;R}\pa{t} {\frac{\partial \mu^{\star}_{l'l;R}}{\partial R}}\pa{t}\mathrm{d}t,\\
   M &\triangleq \frac{1}{4 N_{\r} N_{\t} \mathcal{E}_c} \sum_{l'=1}^{N_{\r}}\sum_{l=1}^{N_{\t}}\int \mu_{l'l;R}\pa{t}{\frac{\partial^2\mu^{\star}_{l'l;R}}{\partial R^2}}\pa{t}\mathrm{d}t.
\end{align*}
It remains to develop the expressions of $M$ and $N$ for model \eqref{eq:model}. After lengthy developments present in the arXiv version of this paper \cite{thiranarXiv}, using the Parseval theorem to translate the time derivatives in the frequency domain and linking them to $f_M$ and $B_{\text{RMS}}$, the proposition is obtained. 
\fi 
\end{proof}
In \eqref{eq:CRB}, one can identify the array gain $N_{\r}N_{\t}$ which multiplies the SNR. This expression also highlights the cooperation existing between the NF phase term $(f_c+f_M)^2(\eta-\beta^2)$, and the waveform term $\eta B_{\text{RMS}}^2$, similarly to what has been observed for the ambiguity functions. The NF effect depends on the carrier frequency and waveform central frequency while the waveform information depends on its RMS bandwidth. The NF effect also depends on the array geometry through the $\sqrt{\eta-\beta^2}$ term, which is the RMS value of the propagation distance derivatives, since it can be written as 
\begin{align}
    \eta-\beta^2 = \frac{1}{N_{\r} N_{\t}}\sum_{l'=1}^{N_{\r}}\sum_{l=1}^{N_{\t}} \pa{\frac{1}{2}\frac{\partial r_{l'l}}{\partial R}-\beta}^2. 
\end{align}
While significant in the NF region, this term goes to $0$ in the FF region of an array, as all propagation distance derivatives become identical when $R$ is large. \Cref{fig:CRB} illustrates the range CRB for an ET in the MIMO case, for two bandwidths and carrier frequencies. One can observe that for small ranges, the NF information is dominant, depending on the carrier frequency. For large ranges, the NF effect vanishes compared to the waveform information. Therefore, the CRB saturates to the waveform limit.  %As a side note, it should noted that the factors $2$ in the definition of $\eta$ and $\beta$ comes from the fact that $r_{l'l}$ includes two times the distance $R$ since it is the total propagation distance. To conclude this discussion, it is again emphasised the above CRB results have been obtained without the Fresnel approximation, unlike the ambiguity function expressions. 

\begin{figure}[h]
    \centering
    \includegraphics[width=\columnwidth]{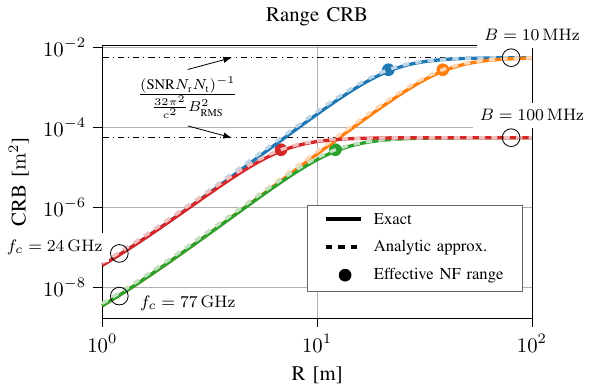}
    \vspace{-0.8cm}
    \caption{CRB for the ET MIMO case, for two bandwidths and carrier frequencies, with $N_{\t} = N_{\r} = 25$, $D=\SI{1.5}{m}$, and $\text{SNR} = \SI{10}{dB}$. The corresponding Rayleigh distances are respectively $\SI{360}{m}$ and $\SI{1155}{m}$. The considered waveform is a cardinal sine, for which $f_M=\SI{0}{Hz}$ and $B_{\text{RMS}} = B/\sqrt{12}$. The dashed lines correspond to the analytical expressions of \Cref{prop:analytic_CRB}. The circles depict the ranges at which the NF information vanishes with respect to the waveform one, as given by \eqref{eq:R_trans}.}
    \label{fig:CRB}
\end{figure}

\subsection{Point and extended target models}
As for the ambiguity functions, analytic expressions for $\eta$ and $\beta$ are obtained in \Cref{prop:analytic_CRB} for the system model of \Cref{fig:sys_model}, to gain insights on the parameters' impact.

\begin{proposition}
    \label{prop:analytic_CRB}
    With \Cref{ass:N_ant}, the factors $\eta$ and $\beta$ only depend on $u\triangleq R/D$ as follows:
    \begin{align*}
        &\left\{\begin{array}{rl}
            \normalfont\eta^{\text{PT,SIMO}} &\normalfont\approx \frac{1}{4}+\frac{u}{2} \text{atan}\pa{\frac{1}{2u}}+u \:\text{asinh}\pa{\frac{1}{2u}},\\[5pt]
            \normalfont\beta^{\text{PT,SIMO}} &\normalfont\approx \frac{1}{2} + u \:\text{asinh}\pa{\frac{1}{2u}},
        \end{array}\right.\\
        &\left\{\begin{array}{rl}
            \normalfont\eta^{\text{PT,MIMO}} &\normalfont\approx u \:\text{atan}\pa{\frac{1}{2u}} +2 u^2 \text{asinh}^2\pa{\frac{1}{2u}},\\[5pt]
            \normalfont\beta^{\text{PT,MIMO}} &\normalfont\approx 2u \:\text{asinh}\pa{\frac{1}{2u}},
        \end{array}\right.\\
        &\left\{\begin{array}{rl}
            \normalfont \eta^{\text{ET,SIMO}}  &\normalfont \approx 4u\: \text{atan}\pa{\frac{1}{4u}},\\[5pt]
            \normalfont\beta^{\text{ET,SIMO}}  &\normalfont \approx 4u \:\text{asinh}\pa{\frac{1}{4u}},
        \end{array}\right.\\
        &\left\{\begin{array}{rl}
            \normalfont\eta^{\text{ET,MIMO}}  &\normalfont \approx 4u\:\text{atan}\pa{\frac{1}{2u}}-4u^2\log\pa{1+\frac{1}{4u^2}},\\[5pt]
            \normalfont\beta^{\text{ET,MIMO}}  &\normalfont \approx 4u\:\text{asinh}\pa{\frac{1}{2u}}-8u^2\pa{\sqrt{1+\frac{1}{4u^2}}-1}.
        \end{array}\right.
    \end{align*}
\end{proposition}
\begin{proof}
    \iffullproof
    In the case of a PT, the range derivative is given by   
\begin{align*}
    \frac{\partial r^{\text{PT}}_{l'l}}{\partial R} = \frac{R}{\sqrt{R^2+z_l^2}}  + \frac{R}{\sqrt{R^2+y_{l'}^2}}.
\end{align*}
Thus,
\begin{align*}
    \eta^{\text{PT}} 
    &= \frac{1}{4N_{\t}}\sum_{l=1}^{N_{\t}}\frac{1}{1+\pa{\frac{z_l}{R}}^2}  + \frac{1}{4N_{\r}}\sum_{l'=1}^{N_{\r}}\frac{1}{1+\pa{\frac{y_{l'}}{R}}^2}\nonumber \\ &\quad + \frac{1}{4N_{\r} N_{\t}}\sum_{l=1}^{N_{\t}}\sum_{{l'}=1}^{N_{\r}} \frac{2}{\sqrt{1+\pa{\frac{z_l}{R}}^2}\sqrt{1+\pa{\frac{y_{l'}}{R}}^2}},\\
    \beta^{\text{PT}}  &= \frac{1}{2N_{\t}}\sum_{l=1}^{N_{\t}}\frac{1}{\sqrt{1+\pa{\frac{z_l}{R}}^2}} + \frac{1}{2N_{\r}}\sum_{{l'}=1}^{N_{\r}}\frac{1}{\sqrt{1+\pa{\frac{y_{l'}}{R}}^2}},
\end{align*}
When the number of antennas is large, the above sums are approximated as 
\begin{align*}
    \frac{1}{4N_{\t}}\sum_{l=1}^{N_{\t}}\frac{1}{1+\pa{\frac{z_l}{R}}^2}  &\approx \frac{R}{2D}\int_{0}^{\frac{D}{2R}}\frac{1}{1+u^2} \mathrm{d}u,\\
    & = \frac{R}{2D} \text{atan}\pa{\frac{D}{2R}},\\
    \frac{1}{2N_{\t}}\sum_{l=1}^{N_{\t}}\frac{1}{\sqrt{1+\pa{\frac{z_l}{R}}^2}} &\approx \frac{R}{D}\int_{0}^{\frac{D}{2R}}\frac{1}{\sqrt{1+u^2}} \mathrm{d}u,\\
    & = \frac{R}{D} \text{asinh}\pa{\frac{D}{2R}},
\end{align*}
providing the desired result. In the case of an extended target model, the range derivative is given by
\begin{align*}
    \frac{\partial r^{\text{ET}}_{l'l}}{\partial R} = \frac{2}{\sqrt{1+\pa{\frac{z_l-y_{l'}}{2R}}^2}}.
\end{align*}
Thus,
\begin{align*}
    \eta^{\text{ET,SIMO}} &= \frac{1}{N_{\r}}\sum_{l'=1}^{N_{\r}}\frac{1}{1+\pa{\frac{y_{l'}}{2R}}^2}
    %&\approx \frac{4R}{D}\int_{0}^{\frac{D}{4R}}\frac{1}{\sqrt{1+z^2}}\mathrm{d}z,\\
    \approx \frac{4 R}{D}\text{atan}\pa{\frac{D}{4R}},\\
    \beta^{\text{ET,SIMO}} &= \frac{1}{N_{\r}}\sum_{l'=1}^{N_{\r}}\frac{1}{\sqrt{1+\pa{\frac{y_{l'}}{2R}}^2}}
    %&\approx \frac{4R}{D}\int_{0}^{\frac{D}{4R}}\frac{1}{\sqrt{1+z^2}}\mathrm{d}z,\\
    \approx \frac{4R}{D}\text{asinh}\pa{\frac{D}{4R}},
\end{align*}
approximating the sum as an integral under \Cref{ass:N_ant}. In the MIMO case, 
\begin{align*}
    \eta^{\text{ET,MIMO}} &= \frac{1}{N_{\r} N_{\t}}\sum_{l'=1}^{N_{\r}}\sum_{l=1}^{N_{\t}}\frac{1}{1+\pa{\frac{z_l-y_{l'}}{2R}}^2},\\
    &\approx \frac{4R^2}{D^2}\int_{-\frac{D}{4R}}^{\frac{D}{4R}}\int_{-\frac{D}{4R}}^{\frac{D}{4R}}\frac{1}{1+\pa{u-v}^2}\mathrm{d}u\mathrm{d}v,\\
    %&= \frac{4R^2}{D^2}\int_{-\frac{D}{4R}}^{\frac{D}{4R}}\int_{-\frac{D}{4R}-v}^{\frac{D}{4R}-v}\frac{1}{1+\pa{t}^2}\mathrm{d}t\mathrm{d}v,\\
    %&= \frac{4R^2}{D^2}\int_{-\frac{D}{4R}}^{\frac{D}{4R}}\text{atan}\pa{\frac{D}{4R}-v}-\text{atan}\pa{-\frac{D}{4R}-v}  \mathrm{d}v,\\
   % &= -\frac{4R^2}{D^2}\int_{-\frac{D}{2R}}^{0}\text{atan}\pa{w}\mathrm{d}w+\frac{4R^2}{D^2}\int_{0}^{\frac{D}{2R}}\text{atan}\pa{w}  \mathrm{d}w,\\
    &= \frac{4R}{D}\text{atan}\pa{\frac{D}{2R}}-\frac{4R^2}{D^2} \log\pa{1+\pa{\frac{D}{2R}}^2},
\end{align*} 
\begin{align*}
    \beta^{\text{ET,MIMO}} &= \frac{1}{N_{\r} N_{\t}}\sum_{l'=1}^{N_{\r}}\sum_{l=1}^{N_{\t}}\frac{1}{\sqrt{1+\pa{\frac{z_l-y_{l'}}{2R}}^2}},\\
    &\approx \frac{4R^2}{D^2}\int_{-\frac{D}{4R}}^{\frac{D}{4R}}\int_{-\frac{D}{4R}}^{\frac{D}{4R}}\frac{1}{\sqrt{1+\pa{u-v}^2}}\mathrm{d}u\mathrm{d}v,\\
    %&= \frac{4R^2}{D^2}\int_{-\frac{D}{4R}}^{\frac{D}{4R}}\int_{-\frac{D}{4R}-v}^{\frac{D}{4R}-v}\frac{1}{\sqrt{1+\pa{t}^2}}\mathrm{d}t\mathrm{d}v,\\
    %&= \frac{4R^2}{D^2}\int_{-\frac{D}{4R}}^{\frac{D}{4R}}\text{asinh}\pa{\frac{D}{4R}-v}-\text{asinh}\pa{-\frac{D}{4R}-v}  \mathrm{d}v,\\
    %&= -\frac{4R^2}{D^2}\int_{-\frac{D}{2R}}^{0}\text{asinh}\pa{w}\mathrm{d}w+\frac{4R^2}{D^2}\int_{0}^{\frac{D}{2R}}\text{asinh}\pa{w}  \mathrm{d}w,\\
    &=  \frac{4R}{D}\text{asinh}\pa{\frac{D}{2R}}-\frac{8R^2}{D^2}\pa{\sqrt{1+\frac{D^2}{4R^2}}-1},
\end{align*}
concluding the proof.
\else
The proposition is obtained by computing the derivatives of the propagation distances, and then approximating the sums as integrals which have the above closed-form expressions. 
\fi
\end{proof}
The above shows that the ratio $R/D$ is the key parameter for the CRB. As it can be observed in \Cref{fig:CRB}, the analytic approximations are in very close agreement with the exact CRB expressions even for low antenna densities, e.g., $25$ antennas over $\SI{1.5}{m}$ in this case. Focusing on the NF geometry term $\eta-\beta^2$, the analytic expressions of \Cref{prop:analytic_CRB} are depicted in \Cref{fig:NF_CRB_term}. As it can be observed, while the above functions are intricate, the CRB term evolves almost linearly, in log-scale, when $R\geq D$. This observation is formalised in the below corollary, which provides a Taylor series expansion of the CRB for $R\geq D$.
\begin{corollary}
    If \Cref{ass:RgeqD} holds,
    \begin{align}
        \normalfont\text{CRB}_{R} & \normalfont \approx  \frac{\pa{N_{\r} N_{\t}\text{SNR}}^{-1}}{\frac{32\pi^2}{c^2} \pa{\frac{\alpha \pa{f_c+f_M}^2}{11520}\pa{\frac{D}{R}}^4 +B_{\text{RMS}}^2  }},
    \end{align}
    with the factor $\alpha = 4$ (resp. $8$, $1$, $7$) in the PT SIMO case (resp. PT MIMO, ET SIMO, ET MIMO). 
    \label{cor:series}
\end{corollary}
\begin{proof}
    \iffullproof
    Developping the above expressions with Taylor series around $D/R=0$, one obtains 
    \begin{align*}
        \eta^{\text{PT,SIMO}} - \pa{\beta^{\text{PT,SIMO}}}^2 &= \frac{D^4}{2880 R^4} + \mathcal{O}\pa{\frac{D^6}{R^6}},\\
        \eta^{\text{PT,SIMO}} &= 1+\mathcal{O}\pa{\frac{D^2}{R^2}}.
    \end{align*}
    Under \Cref{ass:RgeqD}, the higher order terms are negligible, concluding the proof in the SIMO case. The proof is performed similarly in the MIMO case, leading to
    \begin{align*}
        \eta^{\text{PT,SIMO}} - \pa{\beta^{\text{PT,SIMO}}}^2 &= \frac{2D^4}{2880 R^4} + \mathcal{O}\pa{\frac{D^6}{R^6}},\\
        \eta^{\text{PT,MIMO}} &= 1+\mathcal{O}\pa{\frac{D^2}{R^2}}.
    \end{align*}
    Again developping the above with Taylor series, the following is obtained
\begin{align*}
    \eta^{\text{ET,SIMO}} - \pa{\beta^{\text{ET,SIMO}}}^2 &= \frac{D^4}{11520 R^4} + \mathcal{O}\pa{\frac{D^6}{R^6}},\\
    \eta^{\text{ET,SIMO}} &= 1+\mathcal{O}\pa{\frac{D^2}{R^2}},\\
    \eta^{\text{ET,MIMO}} - \pa{\beta^{\text{ET,MIMO}}}^2 &= \frac{7D^4}{11520 R^4} + \mathcal{O}\pa{\frac{D^6}{R^6}},\\
    \eta^{\text{ET,MIMO}} &= 1+\mathcal{O}\pa{\frac{D^2}{R^2}}.
\end{align*}
\else
    The corollary is obtained by developping the expressions of \Cref{prop:analytic_CRB} for $\eta-\beta^2$ and $\eta$ alone, in Taylor series around $D/R=0$. Under \Cref{ass:RgeqD}, only the lowest power of $D/R$ is kept, concluding the proof.
    \fi
\end{proof}

% \begin{corollary}
%     If furthermore \Cref{ass:RgeqD} holds, then
%     \begin{align}
%         \normalfont\text{CRB}^{\text{PT,SIMO}}_{R} &\normalfont  \approx  \frac{1}{8\text{SNR} N_{\r} \pa{\frac{4k^2}{11520}\pa{\frac{D}{R}}^4 +\frac{B_{\text{RMS}}^2}{c^2}  }},\\
%         \normalfont\text{CRB}^{\text{PT,MIMO}}_{R} & \normalfont \approx  \frac{1}{8\text{SNR} N_{\t} N_{\r} \pa{\frac{8 k^2}{11520}\pa{\frac{D}{R}}^4 +\frac{B_{\text{RMS}}^2}{c^2}  }}.
%     \end{align}
%     \label{cor:PT}
% \end{corollary}

\begin{figure}[h]
    \centering
    \includegraphics[width=\columnwidth]{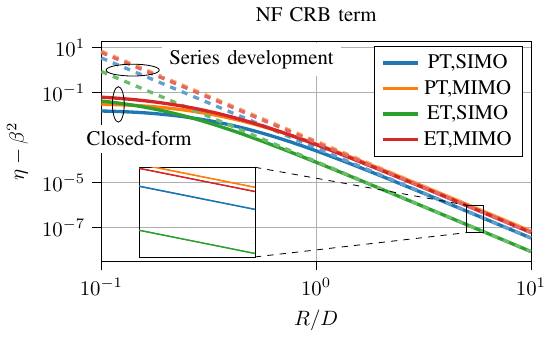}
    \vspace{-0.8cm}
    \caption{Analytic expressions of \Cref{prop:analytic_CRB} for the NF term $\eta-\beta^2$ of the range CRB, as a function of $\frac{R}{D}$. The series developments of \Cref{cor:series} are depicted with dashed lines. }
    \label{fig:NF_CRB_term}
\end{figure}
The above corollary enables to obtain simple closed-form expressions of the NF CRBs for $R\geq D$. In all cases (SIMO and MIMO, PT and ET), the NF effect evolves in $(D/R)^4$. Yet, a constant vertical offset differentiates the cases, which mirrors the difference in curvatures observed in \Cref{fig:NF_amb}. This offset depends on the value of the factor $\alpha$: in order to achieve the same CRB as the one of a MIMO array of size $D$ with a PT (i.e. the best case), the array size must be equal to $D\sqrt[4]{8/\alpha}$, leading to $1.03D$ for the MIMO ET for instance. Another insight provided by \Cref{cor:series} is the range at which the NF effect vanishes with respect to the waveform effect, delimiting the effective NF region. Denoting this effective NF range as $R_{\text{NF,eff}}$, it is obtained when the two terms of the CRB denominator are equal, leading to
\begin{align}
    R_{\text{NF,eff}} = D \sqrt{\frac{f_c+f_M}{B_{\text{RMS}}}}\sqrt[4]{\frac{\alpha}{11520}}.
    \label{eq:R_trans}
\end{align}
These ranges are highlighted in \Cref{fig:CRB}, showing their ability to predict the region in which NF gains are significant.

\section{Conclusion}
This paper provides analytic performance bounds for the multi-antenna range estimation of PTs and ETs. The considered performance metrics are the ambiguity function and the CRB. Both metrics reveal a cooperation between the NF phase information and the waveform delay information. For the considered settings, the obtained performance bounds have particularly simple forms from which the parameter impact can be evaluated. Among others, this reveals that the NF CRBs improve with the RMS value of the propagation distance derivatives. Moreover, an expression of the range at which the NF effect vanishes with respect to the waveform one is provided. Future works include the study of other antenna configurations, and performance bounds for the problem of range and angle estimation. 

\bibliographystyle{ieeetr}
\bibliography{main}

%\appendices

%\input{appendices}

%%%%%%%%%%%%%%%%%%%%%%
\end{document}